\documentclass[12pt]{article}

\usepackage{amsmath}
\usepackage{amssymb}
\usepackage{amsfonts}
\usepackage{latexsym}
\usepackage{color}

\catcode `\@=11 \@addtoreset{equation}{section}
\def\theequation{\arabic{section}.\arabic{equation}}
\catcode `\@=12

\newcommand{\be}{\begin{equation}}
\newcommand{\en}{\end{equation}}
\newcommand{\bea}{\begin{eqnarray}}
\newcommand{\ena}{\end{eqnarray}}
\newcommand{\beano}{\begin{eqnarray*}}
\newcommand{\enano}{\end{eqnarray*}}
\newcommand{\bee}{\begin{enumerate}}
\newcommand{\ene}{\end{enumerate}}

\newcommand{\mc}{\mathcal}

\newcommand{\D}{{\mc D}}
\newcommand{\Vc}{{\mc V}}

\newcommand{\Sc}{{\cal S}}
\newcommand{\E}{{\cal E}}
\newcommand{\F}{{\cal F}}
\newcommand{\G}{{\cal G}}

\newcommand{\Lc}{{\cal L}}
\newcommand{\Mc}{{\cal M}}
\newcommand{\C}{{\cal C}}
\newcommand{\1}{1 \!\! 1}

\newcommand{\Hil}{\mc H}

\newtheorem{thm}{Theorem}

\newtheorem{lemma}[thm]{Lemma}
\newtheorem{prop}[thm]{Proposition}
\newtheorem{defn}[thm]{Definition}

\newenvironment{proof}{\noindent {\bf Proof --}}{\hfill$\square$ \vspace{3mm}\endtrivlist}

\catcode `\@=11 \@addtoreset{equation}{section}
\catcode `\@=12

\begin{document}

\thispagestyle{empty}

\vspace*{2cm}

\begin{center}
{\Large \bf $kq$-representation for pseudo-bosons, and completeness of bi-coherent states}   \vspace{2cm}\\

{\large F. Bagarello}\\
  Dipartimento di Energia, Ingegneria dell'Informazione e Modelli Matematici,\\
Scuola Politecnica, Universit\`a di Palermo,\\ I-90128  Palermo, Italy\\
and I.N.F.N., Sezione di Torino\\
e-mail: fabio.bagarello@unipa.it\\
home page: www.unipa.it/fabio.bagarello

\end{center}

\vspace*{2cm}

\begin{abstract}
\noindent We show how the Zak $kq$-representation can be adapted to deal with pseudo-bosons, and under which conditions. Then we use this representation to prove completeness of a discrete set of bi-coherent states constructed by means of pseudo-bosonic operators.
The case of Riesz bi-coherent states is analyzed in detail.
\end{abstract}

\vspace{2cm}


\vfill


\newpage

\section{Introduction}

In a series of papers the notion of $\D$-pseudo bosons ($\D$-PBs) has been introduced and studied in many details. We refer to \cite{baginbagbook} for a recent review on this subject and for more references. In particular, we have analyzed the functional structure arising from two operators $a$ and $b$, acting on a Hilbert space $\Hil$ and satisfying, in a suitable sense, the pseudo-bosonic commutation rule $[a,b]=\1$. Here $\1$ is the identity operator. We have shown how two biorthogonal families of eigenvectors of two non self-adjoint, number-like, operators can easily be constructed, having real eigenvalues, and we have discussed how and when these operators are similar to a single self-adjoint number operator, and which kind of intertwining relations can be deduced. We have also seen that this settings is strongly related to physics, and in particular to $PT$ and to pseudo-hermitian quantum mechanics, \cite{ben,mosta}, since several models originally introduced in those contexts can be written in terms of $\D$-PBs.

More recently, again in connection with $\D$-PBs, the notion of bi-coherent states (BCS), originally introduced in \cite{tri}, has been considered in some of its aspects, see \cite{bagpb1,abg2015,bagbjalowieza}. Here we continue this analysis and, for that, we first introduce a generalized version of the so-called $kq$-representation, originally considered in the context of many-body theory and studied in details in several papers, \cite{zak2,zak3,jan}. We show that, extending what is done for ordinary coherent states \cite{}, this generalized $kq$-representation can be successfully used to prove that certain discrete sets of BCS  are complete in $\Hil$.

This article is organized as follows: in the next section, to keep the paper self-contained, we review few facts on $\D$-PBs. In Section \ref{sectBCS} we introduce our BCS and discuss some of their properties. Section \ref{sectBCS} also contains our generalized $kq$-representation, while its use in connection with BCS is discussed in Section \ref{sectComplBCS}, where we show that it is possible to extract two discrete sets of these vectors which are both complete in $\Hil$. In Section \ref{sectRBCS} the particular case of regular $\D$-PBs and regular BCS is considered. Our conclusions and plans for the future are given in Section \ref{sectconcl}. Our paper also contains two Appendices. In the first one, meant for the reader who is not familiar with the standard $kq$-representation, few facts on this representation are listed. The second is a technical Appendix, where an useful formula needed in the main part of the paper is proved.

\section{Modifying the CCR}

In this Section we review some of the relevant facts arising when the canonical commutation relation (CCR) $[c,c^\dagger]=\1$ are replaced by a similar commutation rule, $[a,b]=\1$, where $b$ is not the adjoint of $a$. In this case, under suitable assumptions which are often verified in concrete models, an interesting functional structure can be recovered, and the notion of coherent states can also be introduced.

\vspace{2mm}

Let $\Hil$ be a given Hilbert space with scalar product $\left<.,.\right>$ and related norm $\|.\|$. We will assume here that $\Hil$ is {\em maximal}, i.e. that, given a vector $f$, in general belonging to some vector space $\Vc$ larger than $\Hil$, if $\left<f,g\right>$ is well defined for all $g\in\Hil$, then $f$ must necessarily be in $\Hil$ as well. This is, for instance, what happens in $\Lc^2(\Bbb R)$, see \cite{rudin}, while it is not true if we consider $\Hil$ to be a closed subspace of a larger Hilbert space $\Hil_{l}$, endowed with the same scalar product $\left<.,.\right>$ of $\Hil$: in this case, in fact, the fact that $\left<f,g\right>$ is well defined for all $g\in\Hil$ does not prevent $f$ to be an element of $\Hil_l$ not necessarily belonging to $\Hil$.

\vspace{2mm}

Let $a$ and $b$ be two operators
on $\Hil$, with domains $D(a)$ and $D(b)$ respectively, $a^\dagger$ and $b^\dagger$ their adjoint, and let $\D$ be a dense subspace of $\Hil$
such that $a^\sharp\D\subseteq\D$ and $b^\sharp\D\subseteq\D$, where $x^\sharp$ is $x$ or $x^\dagger$. Of course, $\D\subseteq D(a^\sharp)$
and $\D\subseteq D(b^\sharp)$.

\begin{defn}\label{def21}
The operators $(a,b)$ are $\D$-pseudo bosonic ($\D$-pb) if, for all $f\in\D$, we have
\be
a\,b\,f-b\,a\,f=f.
\label{A1}\en
\end{defn}

\vspace{2mm}

Our  working assumptions are the following:

\vspace{2mm}

{\bf Assumption $\D$-pb 1.--}  there exists a non-zero $\varphi_{ 0}\in\D$ such that $a\,\varphi_{ 0}=0$.

\vspace{1mm}

{\bf Assumption $\D$-pb 2.--}  there exists a non-zero $\Psi_{ 0}\in\D$ such that $b^\dagger\,\Psi_{ 0}=0$.

\vspace{2mm}

It is obvious that $\varphi_0\in D^\infty(b):=\cap_{k\geq0}D(b^k)$ and that $\Psi_0\in D^\infty(a^\dagger)$, so
that the vectors \be \varphi_n:=\frac{1}{\sqrt{n!}}\,b^n\varphi_0,\qquad \Psi_n:=\frac{1}{\sqrt{n!}}\,{a^\dagger}^n\Psi_0, \label{A2}\en
$n\geq0$, can be defined and they all belong to $\D$ and, as a consequence, to the domains of $a^\sharp$, $b^\sharp$ and $N^\sharp$, where $N=ba$. We further introduce $\F_\Psi=\{\Psi_{ n}, \,n\geq0\}$ and
$\F_\varphi=\{\varphi_{ n}, \,n\geq0\}$.

It is now simple to deduce the following lowering and raising relations:
\be
\left\{
    \begin{array}{ll}
b\,\varphi_n=\sqrt{n+1}\varphi_{n+1}, \qquad\qquad\quad\,\, n\geq 0,\\
a\,\varphi_0=0,\quad a\varphi_n=\sqrt{n}\,\varphi_{n-1}, \qquad\,\, n\geq 1,\\
a^\dagger\Psi_n=\sqrt{n+1}\Psi_{n+1}, \qquad\qquad\quad\, n\geq 0,\\
b^\dagger\Psi_0=0,\quad b^\dagger\Psi_n=\sqrt{n}\,\Psi_{n-1}, \qquad n\geq 1,\\
       \end{array}
        \right.
\label{A3}\en as well as the eigenvalue equations $N\varphi_n=n\varphi_n$ and  $N^\dagger\Psi_n=n\Psi_n$, $n\geq0$. In particular, as a consequence
of these last two equations,  choosing the normalization of $\varphi_0$ and $\Psi_0$ in such a way $\left<\varphi_0,\Psi_0\right>=1$, we deduce that
\be \left<\varphi_n,\Psi_m\right>=\delta_{n,m}, \label{A4}\en
 for all $n, m\geq0$. Hence $\F_\Psi$ and $\F_\varphi$ are biorthogonal. The analogy with ordinary bosons suggests us to consider the following:

\vspace{2mm}

{\bf Assumption $\D$-pb 3.--}  $\F_\varphi$ is a basis for $\Hil$.

\vspace{1mm}

This is equivalent to requiring that $\F_\Psi$ is a basis for $\Hil$ as well, \cite{chri}. However, several  physical models suggest to adopt the following weaker version of this assumption, \cite{baginbagbook}:

\vspace{2mm}

{\bf Assumption $\D$-pbw 3.--}  For some subspace $\G$ dense in $\Hil$, $\F_\varphi$ and $\F_\Psi$ are $\G$-quasi bases.

\vspace{2mm}
This means that, for all $f$ and $g$ in $\G$,
\be
\left<f,g\right>=\sum_{n\geq0}\left<f,\varphi_n\right>\left<\Psi_n,g\right>=\sum_{n\geq0}\left<f,\Psi_n\right>\left<\varphi_n,g\right>,
\label{A4b}
\en
which can be seen as a weak form of the resolution of the identity, restricted to $\G$.
To refine further the structure, in \cite{baginbagbook} we have assumed that a self-adjoint, invertible, operator $\Theta$ exists, which leaves, together with $\Theta^{-1}$, $\D$ invariant: $\Theta\D\subseteq\D$, $\Theta^{-1}\D\subseteq\D$. Then we say that $(a,b^\dagger)$ are $\Theta-$conjugate if $af=\Theta^{-1}b^\dagger\,\Theta\,f$, for all $f\in\D$. This extends what happens for CCR, where $b=a^\dagger$ and $\Theta=\1$. One can prove that, if $\F_\varphi$ and $\F_\Psi$ are $\D$-quasi bases for $\Hil$, then the operators $(a,b^\dagger)$ are $\Theta-$conjugate if and only if $\Psi_n=\Theta\varphi_n$, for all $n\geq0$. Moreover, if $(a,b^\dagger)$ are $\Theta-$conjugate, then $\left<f,\Theta f\right>>0$ for all non zero $f\in \D$.

In Section \ref{sectRBCS}, rather than using Assumption $\D$-pbw 3 above, we will consider its following stronger version:

\vspace{2mm}

{\bf Assumption $\D$-pbs 3.--}  $\F_\varphi$ is a Riesz basis for $\Hil$.

\vspace{1mm}

In this case we call our $\D$-PBs {\em regular}\footnote{Notice that the $w$ and $s$ in the Assumptions $\D$-pbs 3 and $\D$-pbw 3 stand for {\em strong} and {\em weak}, respectively. Of course, when Assumption $\D$-pbs 3 holds, Assumptions $\D$-pb 3 and $\D$-pbw3 hold as well.}.  In this case a bounded operator $S$, with bounded inverse $S^{-1}$, exists in $\Hil$, together with an orthonormal basis $\F_e=\{e_n\in\Hil,\,n\geq0\}$, such that $\varphi_n=Se_n$, for all $n\geq0$. Then, because of the uniqueness of the  basis biorthogonal to $\F_\varphi$, it is clear that $\F_\Psi$ is also a Riesz basis for $\Hil$, and that $\Psi_n=(S^{-1})^\dagger e_n$. Hence, an operator $\Theta$ having the properties required above can be introduced as $\Theta:=(S^\dagger S)^{-1}$, at least if $\D$ is stable under the action of both $S$ and $S^{-1}$. It is clear that $\Theta$ is also bounded, with bounded inverse,  self-adjoint, positive, and that $\Psi_n=\Theta \varphi_n$, for all $n\geq0$. $\Theta$ and $\Theta^{-1}$ can both be written as a series of rank-one operators. In fact, adopting the Dirac bra-ket notation, we have
$$
\Theta=\sum_{n=0}^\infty |\Psi_n\left>\right<\Psi_n|,\qquad \Theta^{-1}=\sum_{n=0}^\infty |\varphi_n\left>\right<\varphi_n|.
$$
Of course both $|\Psi_n\left>\right<\Psi_n|$ and $|\varphi_n\left>\right<\varphi_n|$ are not projection operators\footnote{Here $\left(|f\left>\right<f|\right)g=\left<f,g\right>f$, for all $f,g\in\Hil$.} since, in general the norms of $\Psi_n$ and $\varphi_n$ are not equal to one. The series above are uniformly convergent if Assumption $\D$-pbs 3 is satisfied, while they are not, if its weaker versions, Assumption $\D$-pb 3 or $\D$-pbw 3, hold.

In several explicit models both $S^\sharp$ and $(S^{-1})^\sharp$ map $\D$ into $\D$,  and this is the reason why we have assumed this condition  here. Hence, $e_n\in\D$, for all $n$. In \cite{baginbagbook} it has also been discussed that an operator $S$ and an orthonormal basis $\F_e$ can also be introduced when our $\D$-PBs satisfy Assumption $\D$-pbw 3. In this case, however, $S$ or $S^{-1}$, or both, are unbounded.

The lowering and raising conditions in (\ref{A3}) for $\varphi_n$ can be rewritten in terms of $e_n$ as follows:
\be
S^{-1}aSe_n=\sqrt{n}\,e_{n-1},\qquad S^{-1}bSe_n=\sqrt{n+1}\,e_{n+1},
\label{a5a}\en
for all $n\geq0$. Notice that we are putting here $e_{-1}\equiv0$. Then, the first equation in (\ref{a5a}) suggests to define an operator $c$ acting on $\D$ as follows: $cf=S^{-1}aSf$. Of course, if we take $f=e_n$, we recover (\ref{a5a}). Moreover, simple computations show that $c^\dagger$ satisfies the equality $c^\dagger f=S^{-1}bS f$, $f\in\D$, which now, taking $f=e_n$, produces the second equality in (\ref{a5a}). These operators satisfy the CCR on $\D$: $[c,c^\dagger]f=f$, $\forall f\in\D$. The conclusion is that the operators $a$ and $b$ are related to a canonical pair $c$ and $c^\dagger$ by a similarity map $S$ on $\D$, which could be bounded together with its inverse, or not. We refer to \cite{baginbagbook} for several applications to physics of this framework.

\section{Bi-coherent states}\label{sectBCS}

We start recalling that, calling $W(z)=e^{zc^\dagger-\overline{z}\,c}$, a {\em standard} coherent state is the vector
\be
\Phi(z)=W(z)e_0=e^{-|z|^2/2}\sum_{k=0}^\infty \frac{z^k}{\sqrt{k!}}\,e_k.
\label{30}\en
Here $c$ and $c^\dagger$ are operators satisfying the CCR, and $\F_e$ is the orthonormal basis related to these operators in the usual way: $c\,e_0=0$, and $e_n=\frac{1}{\sqrt{n!}}\,(c^\dagger)^ne_0$, $n\geq0$. The vector $\Phi(z)$ is well defined, and normalized, for all $z\in\Bbb C$. This is just a consequence of the fact that $W(z)$ is unitary, or, alternatively, of the fact that $\left<e_k,e_l\right>=\delta_{k,l}$. Moreover,
\be
c\,\Phi(z)=z\,\Phi(z),\qquad\mbox{and}\qquad \frac{1}{\pi}\int_{\Bbb C}d^2z|\Phi(z)\left>\right<\Phi(z)|=\1.
\label{add1}\en
It is also well known that $\Phi(z)$ saturates the Heisenberg uncertainty relation, which will not be discussed in this paper.

What is interesting to us here is to show that the family of vectors $\{\Phi(z),\,z\in\Bbb C\}$ can be somehow generalized in a way that preserves similar properties, and that this generalization is related to the $\D$-pb operators $a$ and $b$ considered in the previous section.

Roughly speaking, due to the relation between $(c,c^\dagger)$ with $(a,b)$ or with $(b^\dagger,a^\dagger)$, we expect we can replace $W(z)$ with one of the following operators:
\be
U(z)=e^{zb-\overline{z}\,a},\qquad V(z)=e^{za^\dagger-\overline{z}\,b^\dagger}.
\label{31}\en
Of course, if $a=b^\dagger$, then $U(z)=V(z)$ and the operator is unitary and essentially coincide with $W(z)$, identifying $a$ with $c$. However, the case of interest here is when $a\neq b^\dagger$. This makes the situation more complicated since, in this case, neither $U(z)$ nor $V(z)$ are bounded, in general, at least when $z\neq0$. Still, in \cite{bagpb1,abg2015}, we have found conditions for the vectors
\be
\varphi(z)=U(z)\varphi_0,\qquad \Psi(z)=V(z)\,\Psi_0,
\label{32}\en
to be well defined in $\Bbb C$. This, of course, only means that $\varphi_0$ belongs to the domain of $U(z)$, $\varphi_0\in D(U(z))$, and that $\Psi_0\in D(V(z))$, for all $z\in\Bbb C$. This was proven under some assumptions on the norms of $\varphi_n$ and $\Psi_n$, see Proposition \ref{prop1} below. What we will do here is to check that these same assumptions make $U(z)$ and $V(z)$ densely defined, for all $z\in \Bbb C$. For that, we first recall our main existence result, which can be found, with some differences, in \cite{abg2015}:

\begin{prop}\label{prop1}
Let us assume that there exist four constants $r_\varphi, r_\psi>0$, and $0\leq \alpha_\varphi,\alpha_\psi<\frac{1}{2}$, such that $\|\varphi_n\|\leq r_\varphi^n(n!)^{\alpha_\varphi}$ and $\|\Psi_n\|\leq r_\psi^n(n!)^{\alpha_\psi}$,  for all $n\geq0$.

Then, for all $z\in \Bbb C$, $\varphi_0\in D(U(z))$ and $\Psi_0\in D(V(z))$. Moreover, $\varphi(z)\in D(a)$,  $\Psi(z)\in D(b^\dagger)$, and we have $a\,\varphi(z)=z\varphi(z)$ and $b^\dagger\Psi(z)=z\,\Psi(z)$, for all $z\in \Bbb C$. Finally, if $\F_\varphi$ and $\F_\Psi$ are biorthogonal bases for $\Hil$, then
\be
\left<f,g\right>=\frac{1}{\pi}\int_{\Bbb C}d^2z \left<f,\varphi(z)\right>\left<\Psi(z),g\right>=\frac{1}{\pi}\int_{\Bbb C}d^2z \left<f,\Psi(z)\right>\left<\varphi(z),g\right>,
\label{add2}\en
for all $f,g\in\Hil$. If $\F_\varphi$ and $\F_\Psi$ are $\D$-quasi bases, then equation (\ref{add2}) still holds, but for $f,g\in\D$.

\end{prop}

The proof of the first statement is based on the uniform convergence of the series $\sum\frac{z^k}{\sqrt{k!}}\varphi_k$ and $\sum\frac{z^k}{\sqrt{k!}}\psi_k$, which is granted by the above bounds for $\|\varphi_n\|$ and $\|\Psi_n\|$. The other statements are easy to prove.

\vspace{2mm}

{\bf Remarks:--} (1) This Proposition extends, for $\varphi(z)$ and $\Psi(z)$, similar properties very well known for $\Phi(z)$.

(2) In several concrete models coming from $PT$-quantum mechanics we have seen that $\|\varphi_n\|$ and $\|\Psi_n\|$ diverge with $n$. This prevents the sets $\F_\varphi$ and $\F_\Psi$ to be bases for $\Hil$, but they can be still used to construct BCS as in (\ref{32}) because the divergences of the norms is slower than that admitted in Proposition \ref{prop1}, \cite{abg2015}.

(3) In \cite{bagintop2016} the extension of these BCS to the non-linear situation has also been considered. In this case, the set of natural numbers is replaced by a more general set $\{\epsilon_n\}$, with  $0=\epsilon_0<\epsilon_1<\epsilon_2<\cdots$. The main difference is that the bounds required to the norms of $\varphi_n$ and $\Psi_n$ are now replaced by $\|\varphi_n\|\leq r_\varphi^n(\epsilon_n!)^{\alpha_\varphi}$ and $\|\Psi_n\|\leq r_\Psi^n(\epsilon_n!)^{\alpha_\Psi}$, where $\epsilon_0!=1$ and $\epsilon_n!=\epsilon_1\epsilon_2\cdots\epsilon_n$, $n\geq1$. In this case the BCS can only be defined in some suitable region in $\Bbb C$, rather than in the whole complex plane as in Proposition \ref{prop1}. Also, a resolution of the identity can be found if a certain moment problem related to the $\epsilon_n$'s can be solved. However, this non-linear version of BCS is not relevant for us, since these states are not, in general, eigenstates of the pseudo-bosonic number operators. We refer the interested reader to \cite{bagintop2016} for more details.

\vspace{2mm}

Proposition \ref{prop1} only show that $U(z)$ and $V(z)$ can be applied to two different vectors, $\varphi_0$ and $\Psi_0$. But it is important to stress that, in fact, under the same bounds for $\|\varphi_n\|$ and $\|\Psi_n\|$, both operators are densely defined. Of course, when they are bounded, see Section \ref{sectRBCS}, they are defined in all of $\Hil$, but this is not true in general. To see that $D(U(z))$ and $D(V(z))$ are dense in $\Hil$, we first introduce the following sets:
$$
\Lc_\varphi=l.s.\{\varphi_n\}, \qquad \Lc_\Psi=l.s.\{\Psi_n\}.
$$
These sets are both dense in $\Hil$ in any of the Assumption $\D$-pb 3 (pb3, pbs 3 or pbw 3) considered in Section \ref{sectBCS}. Let us now introduce the operators $\sigma_1=e^{\gamma a}$ and $\sigma_2=e^{\gamma b}$, for some $\gamma\in \Bbb C$, as the following formal (for the moment!) series: $\sigma_1=\sum_{k=0}^\infty\frac{(\gamma\, a)^k}{k!}$ and $\sigma_2=\sum_{k=0}^\infty\frac{(\gamma\, b)^k}{k!}$. Then, both these operators are densely defined. More explicitly:

\begin{lemma}
$D(\sigma_1)\supseteq \Lc_\varphi$. Moreover, if $\|\varphi_n\|\leq r_\varphi^n(n!)^{\alpha_\varphi}$ for all $n\geq0$, for some positive $r_\varphi$ and some $\alpha_\varphi\in\left[0,\frac{1}{2}\right[$, then $D(\sigma_2)\supseteq \Lc_\varphi$.
\end{lemma}
\begin{proof}
The first statement is due to the fact that $a$ acts as a lowering operator on  the vectors of $\F_\varphi$. Then, $e^{\gamma a}\varphi_n$ is just a sum of $n+1$ contributions for all values of $\gamma$, which is clearly a well defined vector of $\Hil$. Therefore each finite linear combination of the $\varphi_n$'s belongs to $D(\sigma_1)$. Hence $D(\sigma_1)\supseteq \Lc_\varphi$. More than this: $\Lc_\varphi$ is stable under the action of $\sigma_1$: $\sigma_1: \Lc_\varphi\rightarrow \Lc_\varphi$.

The situation is completely different for $\sigma_2$, since $b$ behaves as a raising operator on $\F_\varphi$. In fact, we can check that
$$
\sigma_2\varphi_n=\sum_{k=0}^\infty \frac{\gamma^k}{k!}\sqrt{\frac{(n+k)!}{n!}}\,\varphi_{n+k},
$$
for all $n\geq0$. Using now the bound on $\|\varphi_n\|$ we find that
$$
\|\sigma_2\varphi_n\|\leq \frac{r_\varphi^n}{\sqrt{n!}}\sum_{k=0}^\infty (|\gamma|r_\varphi)^k
\frac{[(n+k)!]^{1/2+\alpha_\varphi}}{k!},
$$
which is surely convergent since $\alpha_\varphi<\frac{1}{2}$. Then $\|\sigma_2f\|<\infty$ for all $f\in\Lc_\varphi$.

\end{proof}

{\bf Remark:--} in a similar way we can prove that $D(e^{\gamma a^\dagger})\supseteq \Lc_\Psi$ and that $D(e^{\gamma b^\dagger})\supseteq \Lc_\Psi$, for all $z\in\Bbb C$, at least if $\|\Psi_n\|\leq r_\Psi^n(n!)^{\alpha_\Psi}$, for some positive $r_\Psi$ and some $\alpha_\Psi\in\left[0,\frac{1}{2}\right[$.

\vspace{3mm}

It is now possible to use this Lemma, and the Baker-Campbell-Hausdorff formula, to define the operators $T_1(\alpha)=e^{i\alpha \hat x}$ and $T_2(\beta)=e^{-i\beta \hat p}$, where, in analogy with the standard definitions, $\hat x$ and $\hat p$ are related to the pseudo-bosonic lowering and raising operators $a$ and $b$  as follows:
\be
\hat xf=\frac{a+b}{\sqrt{2}}\,f,\qquad  \hat pf=\frac{a-b}{\sqrt{2}\,i}\,f,
\label{51}\en
 for all $f\in \D$.
For all $f\in\Lc_\varphi$, and for all real $\alpha$ and $\beta$, we put
\be
e^{i\alpha \hat x} f=e^{-\alpha^2/4}e^{i\alpha/\sqrt{2}b}e^{i\alpha/\sqrt{2}a} f,\qquad
e^{-i\beta \hat p} f=e^{-\beta^2/4}e^{\beta/\sqrt{2}b}e^{-\beta/\sqrt{2}a} f.
\label{52}\en
The order here is important: for what discussed before, both $e^{i\alpha/\sqrt{2}a} f$ and $e^{-\beta/\sqrt{2}a} f$  belong to $\Lc_\varphi$, for all $f\in\Lc_\varphi$. Hence we can act on these vectors with $e^{i\alpha/\sqrt{2}b}$ or with $e^{\beta/\sqrt{2}b}$ getting well defined vectors in $\Hil$.

In a similar way we can prove that the operators $e^{-i\alpha \hat x^\dagger}$ and $e^{i\beta \hat p^\dagger}$ are densely defined, since they are both defined on $\Lc_\Psi$. With a little abuse of language, based on the fact that explicit computations show that
\be
\left<e^{-i\alpha \hat x^\dagger}g,f\right>=\left<g,e^{i\alpha \hat x}f\right>, \quad \left<e^{i\beta \hat p^\dagger}g,f\right>=\left<g,e^{-i\beta \hat p}f\right>,
\label{add4}\en
for all $f\in \Lc_\varphi$ and $g\in\Lc_\Psi$, we will call $e^{-i\alpha \hat x^\dagger}$ and $e^{i\beta \hat p^\dagger}$ the {\em adjoints} of $T_1(\alpha)$ and $T_2(\beta)$, and we will simply write $T_1^\dagger(\alpha)=e^{-i\alpha \hat x^\dagger}$ and $T_2^\dagger(\beta)=e^{i\beta \hat p^\dagger}$.

\subsection{An interlude: generalized eigenstates of $\hat x$}\label{sectkqrepr}

We recall that, calling $\hat x_0$ the self adjoint position operator $\hat x_0=\frac{c+c^\dagger}{\sqrt{2}}$, it is possible to introduce its {\em generalized eigenstates} $\xi_x$ satisfying $\hat x_0\,\xi_x=x\,\xi_x$, $\left<\xi_x,\xi_y\right>=\delta(x-y)$, and $\int_{\Bbb R}dx|\xi_x\left>\right<\xi_x|=\1$, using a standard notation.  Of course, $\xi_x$ is not a square integrable function, but can be seen as an element of $\Sc'(\Bbb R)$, the space of tempered distributions on $\Bbb R$, \cite{gieres}.

The role of $\xi_x$ is quite important in kq-representation and, more in general, in representation theory of quantum mechanics, \cite{mess,merz}. For this reason, it is interesting to see what can be extended to our situation. To achieve this aim, we will first propose a general settings, adopting some useful assumptions, and then we will show that these assumptions are indeed satisfied in some relevant situations.

Let $x\in\Bbb R$ labels a tempered distribution $\eta_x\in \Sc'(\Bbb R)$,  and let $\F_\eta$ be the set of all these distributions: $\F_\eta=\{\eta_x,\,x\in\Bbb R\}$.

\begin{defn}\label{defwb}

 $\F_\eta$ is called {\em well-behaved} if:

 1. each $\eta_x$ is a generalized eigenstate of $\hat x$: $\hat x\eta_x=x\eta_x$, for all $x\in\Bbb R$;

 2. a second family of generalized vectors $\F^\eta=\{\eta^x\in\Sc'(\Bbb R),\,x\in \Bbb R\}$ exists such that $\left<\eta_x,\eta^y\right>=\delta(x-y)$ and $\int_{\Bbb R}dx|\eta_x\left>\right<\eta^x|=\int_{\Bbb R}dx|\eta^x\left>\right<\eta_x|=\1$.

\end{defn}

Both $\F_\eta$ and $\F^\eta$ are complete in $\Hil$: if $f\in\Hil$ is such that $\left<\eta_x,f\right>=0$ $\forall x\in \Bbb R$, then since $f=\int_{\Bbb R}dx\left<\eta_x,f\right>\eta^x$, it follows that $f=0$. Analogously, $f=0$ if $\left<\eta^x,f\right>=0$ $\forall x\in \Bbb R$.

\vspace{2mm}

{\bf Remark:--} It is clear that the resolutions of the identity above extend to generalized vectors of $\F_\eta$ and $\F^\eta$. In fact, for instance, we have
$$
\left(\int_{\Bbb R}dx|\eta^x\left>\right<\eta_x|\right)\eta^y=\int_{\Bbb R}dx\left<\eta_x,\eta^y\right>\eta^x=\int_{\Bbb R}dx\,\delta(x-y)\eta^x=\eta^y.
$$
We call $\Mc$ the set of (generalized) vectors for which $\F_\eta$ and $\F^\eta$ produce similar resolutions. For what we have just seen, $\Mc\supset\Hil$. Of course, both $\F_\eta$ and $\F^\eta$ are complete in $\Mc$.

\vspace{2mm}

With this in mind it is possible to prove that, if $\hat x^\dagger\eta^x\in\Mc$,  $\eta^x$ is a generalized eigenstate of $\hat x^\dagger$:
\be
\hat x^\dagger\eta^x=x\,\eta^x,
\label{59}\en
for all $x\in\Bbb R$. In fact, we have
$$
\left<\eta_y,\hat x^\dagger\eta^x\right>=\left<\hat x\eta_y,\eta^x\right>=y\delta(x-y)=x\delta(x-y)=x\left<\eta_y,\eta^x\right>,
$$
so that $\left<\eta_y,\hat x^\dagger\eta^x-x\eta^x\right>=0$ for all $x,y\in\Bbb R$. Then, using completeness of $\F_\eta$ in $\Mc$  we conclude that $\hat x^\dagger\eta^x-x\eta^x=0$.

\vspace{2mm}

{\bf Example:--} As an example we consider the non self-adjoint shifted harmonic oscillator, \cite{baginbagbook,bagharmosc}, where the $\D$-pb operators are defined as $a=c+\alpha\1$ and $b=c^\dagger+\overline{\beta}\,\1$, and the Hamiltonian is $H=ba$ which is clearly non self-adjoint if $\alpha\neq\beta$. Since $\hat x=\frac{a+b}{\sqrt{2}}=\hat x_0+\delta\1$, where $\delta=\frac{\alpha+\overline{\beta}}{\sqrt{2}}$, we see that the eigenstates of $\hat x$ are $\eta_x=\xi_{x-\delta}$, while $\eta^x=\xi_{x-\overline{\delta}}$ are the eigenstates of $\hat x^\dagger$. This is true, of course, if the eigenvalue equation $\hat x_0\,\xi_x=x\,\xi_x$ can be extended (in some non trivial way) to complex $x$. In this case, it is possible to check that $\left<\eta_x,\eta^y\right>=\delta(x-y)$ and that the resolution of the identity holds true: $\int_{\Bbb R}dx\left<f,\eta_x\right>\left<\eta^x,g\right>=\left<f,g\right>$, for $f,g\in\Hil$.

\vspace{2mm}

For each $f\in\Hil$ we can now define two functions as follows:
\be
f^\uparrow(x):=\left<\eta^x,f\right>,\qquad f^\downarrow(x):=\left<\eta_x,f\right>.
\label{53}\en
The first remark is that each pair $(f^\uparrow(x),g^\downarrow(x))$ is {\em compatible}, for any $f,g\in\Hil$: in other words, their product $\overline{f^\uparrow(x)}\,g^\downarrow(x)$ belongs to $\Lc^1(\Bbb R)$. In fact,
$$
\int_{\Bbb R}dx\overline{f^\uparrow(x)}\,g^\downarrow(x)=\int_{\Bbb R}dx\left<f,\eta^x\right>\left<\eta_x,g\right>=\left<f,g\right>,
$$
which is finite since $f,g\in\Hil$. This, of course, does not imply that $f^\uparrow(x)$ and $g^\downarrow(x)$ are both square-integrable. In principle, it could well be that $f^\uparrow(x)\in \Lc^p(\Bbb R)$ and $g^\downarrow(x)\in \Lc^q(\Bbb R)$, with $p^{-1}+q^{-1}=1$. However, it is possible to find a necessary and sufficient condition for both $f^\uparrow(x)$ and $g^\downarrow(x)$ to be square-integrable. For that, we need to introduce the operators $S_\eta$ and $S^\eta$ defined as follows: we introduce first
$$
D(S_\eta)=\left\{f\in\Hil: \int_{\Bbb R}dx\left<\eta_x,f\right>\eta_x\in\Hil \right\}, \quad D(S^\eta)=\left\{f\in\Hil: \int_{\Bbb R}dx\left<\eta^x,f\right>\eta^x\in\Hil \right\},
$$
and then we put
\be
S_\eta f=\int_{\Bbb R}dx\left<\eta_x,f\right>\eta_x, \qquad S^\eta g=\int_{\Bbb R}dx\left<\eta^x,g\right>\eta^x
\label{54}\en
for all $f\in D(S_\eta)$ and $g\in D(S^\eta)$. It is important, in what follows, to require that $D(S_\eta)$ and $D(S^\eta)$ are (at least) dense in $\Hil$.
\vspace{2mm}

{\bf Remark:--} of course, this is the case when $\hat x=\hat x_0$, i.e. when $\D$-PBs are ordinary bosons. In fact, when this is so, $\eta_x=\eta^x=\xi_x$, see Appendix 1, and  $D(S_\eta)=D(S^\eta)=\Hil$. We will see in Section \ref{sectRBCS} that this is not the only case where the density of these sets can be proved.

\vspace{2mm}

It should be remarked that $S_\eta$ and $S^\eta$ can also act outside their domains $D(S_\eta)$ and $D(S^\eta)$ since, for instance,
\be
S_\eta \eta^x=\eta_x,\qquad S^\eta \eta_x=\eta^x,
\label{55}\en
for all $x\in\Bbb R$. Then we have $S_\eta S^\eta\eta_x=\eta_x$, and $S^\eta S_\eta\eta^x=\eta^x$: $S_\eta$ is a {\em sort of inverse}
of $S^\eta$ on $\Mc$. It may be convenient to write these operators in a bra-ket form as
$$
S_\eta=\int_{\Bbb R}dx|\eta_x\left>\right<\eta_x|, \qquad S^\eta=\int_{\Bbb R}dx|\eta^x\left>\right<\eta^x|.
$$
Let us now introduce the sets
$$
\hat D(S_\eta)=\left\{f\in\Hil: \int_{\Bbb R}dx\left|f^\downarrow(x)\right|^2<\infty \right\}, \quad \hat D(S^\eta)=\left\{f\in\Hil: \int_{\Bbb R}dx\left|f^\uparrow(x)\right|^2<\infty \right\}.
$$

These sets are, at a first sight, different from $D(S_\eta)$ and $D(S^\eta)$. However the following results show that this is not so:
\begin{lemma} $f\in D(S_\eta)$ if and only if $ f\in \hat D(S_\eta)$. Analogously, $g\in D(S^\eta)$ if and only if $g\in \hat D(S^\eta)$.
\end{lemma}
\begin{proof}
Using (\ref{54}) it is clear that, if  $f\in D(S_\eta)$, then $\int_{\Bbb R}dx\left|f^\downarrow(x)\right|^2<\infty$, so that $f\in \hat D(S_\eta)$. Hence $D(S_\eta)\subseteq\hat D(S_\eta)$.  To prove the converse we first observe that, if $f, g\in \hat D(S_\eta)$, then $\alpha f+\beta g\in \hat D(S_\eta)$  for all choices of complex $\alpha$ and $\beta$: $\int_{\Bbb R}dx\left|\alpha f^\downarrow(x)+\beta g^\downarrow(x)\right|^2<\infty$. Then since $\int_{\Bbb R}dx\left|f^\downarrow(x)\right|^2=\left<S_\eta\ f, f\right><\infty$, for all $f\in \hat D(S_\eta)$, using the polarization identity we conclude that $\left<S_\eta f,g\right>$ is well defined for all $f,g\in \hat D(S_\eta)$. Now, since $D(S_\eta)$ is dense in $\Hil$, and since we have just shown that $D(S_\eta)\subseteq\hat D(S_\eta)$,  $\hat D(S_\eta)$ is also dense in $\Hil$. Then, using the continuity of the scalar product, we conclude that $\left<S_\eta f,G\right>$ is well defined for all $G\in\Hil$, and for any $f\in \hat D(S_\eta)$. Hence, because of the maximality of $\Hil$, $S_\eta f\in\Hil$. This means that $f\in D(S_\eta)$, as we had to prove.

The equality $D(S^\eta)=\hat D(S^\eta)$ can be proved in a similar way.

\end{proof}

A consequence of this Lemma is the following result

\begin{prop}\label{prop2}
The functions $f^\downarrow(x)\in\Lc^2(\Bbb R)$ for each $f\in\Hil$ if and only if $S_\eta\in B(\Hil)$. Analogously, $f^\uparrow(x)\in\Lc^2(\Bbb R)$ for each $f\in\Hil$ if and only if $S^\eta\in B(\Hil)$.
\end{prop}

\begin{proof}
Let us first assume that $S_\eta\in B(\Hil)$. Then, since
$$
\|f^\downarrow\|^2_{\Lc^2}=\int_{\Bbb R}dx \left<f,\eta_x\right>\left<\eta_x,f\right>=\left<S_\eta f,f\right>\leq \|S_\eta\| \|f\|_{\Hil}^2,
$$
it follows that $f^\downarrow(x)\in\Lc^2(\Bbb R)$. Notice that we are using different notations for the norm of $f$ and $f^\downarrow$, since they live in different spaces.

Let us now suppose that $f^\downarrow(x)\in\Lc^2(\Bbb R)$. Hence, since $\left<S_\eta f,f\right>=\|f^\downarrow\|^2_{\Lc^2}$, it follows that $\left<S_\eta f,f\right>$ is finite for all $f\in \Hil$. Then, using the polarization identity, $\left<S_\eta f,g\right>$ is well defined for all $f,g\in\Hil$. Hence the maximality of $\Hil$ implies that $S_\eta f\in \Hil$, for all $f\in\Hil$. This means that $S_\eta$ is everywhere defined in $\Hil$, so that it must be bounded.

A similar proof can be repeated for our second claim.

\end{proof}

\subsection{Back to the exponential operators}

Let us now go back to the operators in (\ref{52}), and to their adjoints. For our purposes, it will be sufficient to take $\alpha=\beta$, but we will constraint $\alpha$ to satisfy the equality $\alpha^2=2\pi L$, for some $L=1,2,3,\ldots$. In fact, in this case, the two operators $T_1=e^{i\alpha\hat x}$ and $T_2=e^{-i\alpha \hat p}$ commute. Of course, commutation should be understood in the sense of unbounded operators, i.e., as
\be
\left<T_1f,T_2^\dagger g\right>= \left<T_2f,T_1^\dagger g\right>,
\label{56}\en
for all $f\in \Lc_\varphi$ and $g\in\Lc_\Psi$. If, in particular, a set $\mathfrak D$ dense in $\Hil$ exists which is left stable under the action of $T_j^\sharp$, $j=1,2$, then rather than (\ref{56}) we can write $T_1(T_2f)=T_2(T_1f)$, for all $f\in \mathfrak D$.

Of course $T_1$ and $T_2$ admit inverse, and the inverses are $T_1^{-1}=e^{-i\alpha\hat x}$ and $T_2^{-1}=e^{i\alpha \hat p}$, whose (formal) adjoints are $(T_1^{-1})^\dagger=e^{i\alpha\hat x^\dagger}=(T_1^\dagger)^{-1}$ and $(T_2^{-1})^\dagger=e^{-i\alpha \hat p^\dagger}=(T_2^\dagger)^{-1}$, with the same care we used in formula (\ref{add4}) for the adjoints of $T_1$ and $T_2$. The action of these operators on $\eta_x$ and $\eta^x$ can be deduced and it turns out that
\be
T_1\eta_x=e^{i\alpha x}\eta_x, \qquad T_1^\dagger\eta^x=e^{-i\alpha x}\eta^x,
\label{57}\en
with obvious identities for $T_1^{-1}\eta_x$ and $ (T_1^\dagger)^{-1}\eta^x$. These formulas easily follow from  (\ref{59}) and from the eigenvalue equation $\hat x \eta_x=x\,\eta_x$.  As for $T_2^\sharp$, we can check, see Appendix 2, that
\be
T_2\eta_x=\eta_{x+\alpha}, \qquad T_2^\dagger\eta^x=\eta^{x-\alpha},
\label{58}\en
so that $T_2^{-1}\eta_x=\eta_{x-\alpha}$ and $(T_2^\dagger)^{-1}\eta^x=\eta^{x+\alpha}$. From (\ref{57}) and (\ref{58})  standard arguments show that, for all $f\in\Hil$,
\be
\left\{
\begin{array}{ll}
T_1 f^\uparrow(x)=e^{i\alpha x} f^\uparrow(x), \qquad T_2 f^\uparrow(x)= f^\uparrow(x-\alpha), \\
 T_1^\dagger f^\downarrow(x)=e^{-i\alpha x} f^\downarrow(x), \qquad T_2^\dagger f^\downarrow(x)= f^\downarrow(x+\alpha)\\
\end{array}%
\right.
\label{510}
\en
With this in mind, extending what is discussed in the Appendix 1, it is possible to find the generalized eigenstates  of $T_1$ and $T_2$, and of $T_1^\dagger$ and $T_2^\dagger$. If we put
\be
\Phi_{kq}^\uparrow(x)=\frac{1}{\sqrt{\alpha}}\sum_{n\in\Bbb Z}e^{ikn\alpha}\delta(x-q-n\alpha)=\left<\eta^x,\Phi_{kq}\right>,
\label{511}\en
for $k,q\in[0,\alpha[$, then it is clear that they are not square integrable in $\Bbb R$. Still, they have interesting features. In fact, using (\ref{510}) and a bit of algebra, we get
\be
T_1\Phi_{kq}^\uparrow(x)=e^{i\alpha q}\Phi_{kq}^\uparrow(x),\qquad T_2\Phi_{kq}^\uparrow(x)= e^{-i\alpha k}\Phi_{kq}^\uparrow(x).
\label{512}\en
In a similar way, introducing the distributions
\be
\Psi_{kq}^\downarrow(x)=\left<\eta_x,\Psi_{kq}\right>=\Phi_{kq}^\uparrow(x),
\label{513}\en
we find that
\be
T_1^\dagger\Psi_{kq}^\downarrow(x)=e^{-i\alpha q}\Psi_{kq}^\downarrow(x),\qquad T_2^\dagger\Psi_{kq}^\downarrow(x)= e^{i\alpha k}\Psi_{kq}^\downarrow(x).
\label{514}\en

\vspace{2mm}

{\bf Remark:--} It might be worth noticing that, even if $\Psi_{kq}^\downarrow(x)=\Phi_{kq}^\uparrow(x)$, in general $\Psi_{kq}\neq\Phi_{kq}$. This is because $\eta_x$ and $\eta^x$ are different. This will be clarified in Section \ref{sectRBCS}.

\vspace{2mm}

Now, in complete analogy with what stated in the Appendix 1, calling $\Box=\{(k,q)\in{\Bbb R}^2: k,q\in[0,\alpha[\}$, we can check the following results
\be
\int_{\Bbb R} \overline{\Psi^\downarrow_{kq}(x)}\Phi^\uparrow_{k'q'}(x)dx=\int_{\Bbb R} \overline{\Psi^\downarrow_{kq}(x)}\Psi^\downarrow_{k'q'}(x)dx=\int_{\Bbb R} \overline{\Phi^\uparrow_{kq}(x)}\Phi^\uparrow_{k'q'}(x)dx=\delta(k-k')\delta(q-q'),
\label{515}\en
and
\be
\int\int_{\Box} \overline{\Psi^\downarrow_{kq}(x)}\Phi^\uparrow_{kq}(x')dk\,dq=\int\int_{\Box} \overline{\Phi^\uparrow_{kq}(x)}\Psi^\downarrow_{kq}(x')dk\,dq=\delta(x-x').
\label{516}\en
From (\ref{515}) the following equalities can be deduced:
\be
\left<\Psi_{kq},\Phi_{k'q'}\right>=\left<S_\eta\Psi_{kq},\Psi_{k'q'}\right>=\left<S^\eta\Phi_{kq},\Phi_{k'q'}\right>=
\delta(k-k')\delta(q-q'),
\label{517}\en
while from (\ref{516}) we deduce that
\be
\int\int_{\Box}|\Psi_{kq}\left>\right<\Phi_{kq}|dk\,dq=\int\int_{\Box}|\Phi_{kq}\left>\right<\Psi_{kq}|dk\,dq=\1.
\label{518}\en
Since we also have
\be
\int\int_{\Box} \overline{\Psi^\downarrow_{kq}(x)}\Psi^\downarrow_{kq}(x')dk\,dq=\int\int_{\Box} \overline{\Phi^\uparrow_{kq}(x)}\Phi^\uparrow_{kq}(x')dk\,dq=\delta(x-x'),
\label{add3}\en
it is natural to introduce formally\footnote{In fact, these operators could be unbounded, and a certain mathematical care is needed  for a rigorous definition. However, since they will play no role in the rest of the paper, we will skip this point here.} here also the operators
\be
R_\Psi=\int\int_{\Box}|\Psi_{kq}\left>\right<\Psi_{kq}|dk\,dq, \qquad R_\Phi=\int\int_{\Box}|\Phi_{kq}\left>\right<\Phi_{kq}|dk\,dq.
\label{519}\en
Then
\be
\left<\eta^{x'},R_\Phi\,\eta^x\right>=\left<\eta_{x'},R_\Psi\,\eta_x\right>=\delta(x-x').
\label{520}\en
The conclusion is therefore that $(S_\eta,S^\eta)$ and $(R_\Psi,R_\Phi)$ share a similar behavior, in the following sense:
\be
S^\eta\Phi_{kq}=\Psi_{kq},\quad S_\eta\Psi_{kq}=\Phi_{kq}, \quad \mbox{while}\quad R_\Psi\eta_x=\eta^x,\quad R_\Phi\eta^x=\eta_x. \label{521}\en
Then, in a certain sense, $S^\eta$ and $R_\Psi$ can be thought as the inverses of $S_\eta$ and $R_\Phi$, respectively.

We end this section by observing that the equalities in (\ref{512}) and (\ref{514}) can be rewritten as follows:
\be
T_1\Phi_{kq}=e^{i\alpha q}\Phi_{kq},\quad T_2\Phi_{kq}= e^{-i\alpha k}\Phi_{kq},
\quad
T_1^\dagger\Psi_{kq}=e^{-i\alpha q}\Psi_{kq},\quad T_2^\dagger\Psi_{kq}= e^{i\alpha k}\Psi_{kq}.
\label{522}\en

\section{Completeness of BCS}\label{sectComplBCS}

It is well known that standard coherent states $\Phi(z)$ in (\ref{30}) are overcomplete. This is the essence of the resolution of the identity in (\ref{add1}), together with the fact that we can remove out of the set $\{\Phi(z), \,z\in\Bbb C\}$ some vector, $\Phi(z_j)$, $j=1,2,\ldots,N$ for instance, still getting a complete set. In past years, it has been discussed how to take a suitable countable subset of $\Bbb C$, ${\Bbb C}_{num}$, such that the discrete set of coherent states $\{\Phi(z_j),\, z_j\in {\Bbb C}_{num}\}$, is still complete. Results in this directions can be found in \cite{klau,perel,zak}.

We want to show that our previous analysis allows us to answer  a similar question also for our BCS. More explicitly we will show that, calling $z_{\bf n}=\frac{\alpha}{\sqrt{2}}(n_1+in_2)$, where ${\bf n}=(n_1,n_2)$ and $n_1, n_2\in\Bbb Z$, the completeness of the sets $\{\varphi(z), \,z\in\Bbb C\}$ and $\{\Psi(z), \,z\in\Bbb C\}$ following from Proposition \ref{prop1} is still true for the two sets $\C_\varphi=\{\varphi(z_{\bf n}), \,{\bf n}\in {\Bbb Z}^2\}$ and $\C_\Psi=\{\Psi(z_{\bf n}), \,{\bf n}\in {\Bbb Z}^2\}$, at least if $L=1$ (i.e. if $\alpha^2=2\pi$), while is false for $L=2,3,4,\ldots$.

Let us assume, for the moment, that $\alpha^2=2\pi L$. Then $U(z_{\bf n})=(-1)^{Ln_1n_2}T_1^{ n_2}T_2^{ n_1}$. Let $f\in\Hil$ be orthogonal to all the $\varphi(z_{\bf n})$, ${\bf n}\in {\Bbb Z}^2$: $\left<f,U(z_{\bf n})\varphi_0\right>=0$, $\forall\,{\bf n}\in {\Bbb Z}^2$. This implies that $\left<f,T_1^{n_2}T_2^{n_1}\varphi_0\right>=0$, $\forall\, n_j\in\Bbb Z$. Then, because of the (\ref{518}), we have
$$
0=\left<f,T_1^{n_2}T_2^{n_1}\varphi_0\right>=\int\int_{\Box}\left<f,\Phi_{kq}\right>\left<\Psi_{kq},T_1^{n_2}T_2^{n_1}\varphi_0
\right>\,dk\,dq=$$
$$=\int\int_{\Box}\left<f,\Phi_{kq}\right>\left<{T_2^{n_1}}^\dagger{T_1^{n_2}}^{\dagger}\Psi_{kq},\varphi_0
\right>\,dk\,dq=\int\int_{\Box}e^{i\alpha(qn_2-kn_1)}\left<f,\Phi_{kq}\right>\left<\Psi_{kq},\varphi_0\right>\,dk\,dq.
$$
Notice that, if $L=1$, the set $\E:=\{e^{i\alpha(qn_2-kn_1)},\, n_j\in\Bbb Z\}$ is complete in $\Lc^2(\Box)$. Hence we must have $\left<f,\Phi_{kq}\right>\left<\Psi_{kq},\varphi_0\right>=0$.
 Now, $\left<\Psi_{kq},\varphi_0\right>$ cannot be zero almost everywhere (a.e.) in $k$ and $q$ since otherwise $\varphi_0$ would be the zero vector, which is impossible. Hence $\left<f,\Phi_{kq}\right>=0$ a.e. and, consequently, $f=0$, which is what we had to prove. Of course, this conclusion is false if $L=2,3,4,\ldots$, since in this case the set $\E$ is not complete in  $\Lc^2(\Box)$.

Similar arguments can be repeated to prove that, if $L=1$, $\C_\Psi$ is complete in $\Lc^2(\Box)$, while, if $L>1$, it is not.

\section{A particular case: regular $\D$-PBs and Riesz BCS}\label{sectRBCS}

We will now consider the case of regular $\D$-PBs, i.e. the case in which the sets $\F_\varphi$ and $\F_\Psi$ are biorthogonal Riesz bases: then a bounded invertible operator $S$ exists (which we assume, without loss of generality, to be self-adjoint), with $S^{-1}$ bounded, such that $\varphi_n=Se_n$ and $\Psi_n=S^{-1}e_n$, where $\{e_n, \,n\geq0\}$ is an orthonormal basis for $\Hil$. In several explicit models involving $\D$-PBs each $e_n\in\D$ and $\D$ is stable under the action of $S$ and $S^{-1}$. We will work under these simplifying assumptions here. This case has been analyzed in details in \cite{bagbjalowieza}, where, among other things, it was proven that $\varphi(z)$ and $\Psi(z)$ are well defined in $\Hil$ for all $z\in\Bbb{C}$, that $U(z)$ and $V(z)$ are bounded operators and that $\varphi(z)=U(z)\varphi_0=S\Phi(z)$, while $\Psi(z)=V(z)\Psi_0=(S^{-1})^\dagger\Phi(z)$. This means that $(\varphi(z),\Psi(z))$, $z\in\Bbb C$,  are  Riesz BCS (RBCS), see \cite{bagbjalowieza}, so that they satisfy in particular the following properties:

(1) $\forall\,z\in\Bbb C$ $$\left<\varphi(z),\Psi(z)\right>=1.$$

(2) For all $f,g\in\Hil$ the following  resolution of the identity holds:
\be
\left<f,g\right>=\frac{1}{\pi}\int_{\Bbb C}d^2z \left<f,\varphi(z)\right>\left<\Psi(z),g\right>
\label{61}\en

(3) $\forall\,z\in\Bbb C$
\be
a\,\varphi(z)=z\,\varphi(z),\qquad b^\dagger\Psi(z)=z\,\Psi(z)
\label{62}\en

\vspace{2mm}

Our aim in this section is to show that, for regular $\D$-PBs, the assumptions required in Section \ref{sectBCS} are indeed satisfied. First of all, since $\|\varphi_n\|=\|Se_n\|\leq \|S\|$ and $\|\Psi_n\|=\|S^{-1}e_n\|\leq \|S^{-1}\|$, both these norms are uniformly bounded. Hence, the inequalities required in Proposition \ref{prop1} are surely satisfied. This is in agreement with what stated above, that is that $\varphi(z)$ and $\Psi(z)$ are well defined in $\Hil$ for all $z\in\Bbb{C}$. Now, it is useful to observe that, in the present case, calling, as in Section \ref{sectkqrepr}, $\hat x_0=\frac{c+c^\dagger}{\sqrt{2}}$ and $\hat p_0=\frac{c-c^\dagger}{\sqrt{2}\,i}$, these two self-adjoint operators are related to $\hat x$ and $\hat p$ by a similarity transformation:
\be
\hat x f=S\hat x_0 S^{-1}f, \qquad \hat p f=S\hat p_0 S^{-1}f,
\label{63}\en
for all $f\in\D$. Of course, this implies that $\D$ is stable under the action of $\hat x$ and $\hat p$. Suppose now that, in a distributional sense, $\xi_x$ belongs to the domain of $S$ and $S^{-1}$. More explicitly, we assume that, for all $x\in\Bbb R$, $S\xi_x$ and $S^{-1}\xi_x$ both exist and belong to $\Sc'(\Bbb R)$. Then we define
\be
\eta_x=S\xi_x \qquad \mbox{and} \quad \eta^x=S^{-1}\xi_x,
\label{64}\en
and we can easily see that the set $\F_\eta$ is well-behaved in the sense of Definition \ref{defwb}. In fact, assuming that (\ref{63}) can be extended to $\Sc'(\Bbb R)$, we find $\hat x\,\eta_x=S\hat x_0 S^{-1} S\xi_x=S\hat x_0\xi_0=x\,\eta_x$. Also,
$$
\left<\eta_x,\eta^y\right>=\left<S\xi_x,S^{-1}\xi_y\right>=\left<\xi_x,\xi_y\right>=\delta(x-y).
$$
Finally, taken $f\in\Hil$, we have for instance
$$
\int_{\Bbb R}dx\left<\eta_x,f\right>\eta^x=\int_{\Bbb R}dx\left<S\xi_x,f\right>S^{-1}\xi_x=S^{-1}\int_{\Bbb R}dx\left<\xi_x,Sf\right>S^{-1}\xi_x=S^{-1}(Sf)=f,
$$
where the fact that $S^{-1}$ is bounded (and therefore continuous) has been used.

Recalling now our definitions of $f^\uparrow(x)$ and $f^\downarrow(x)$ in (\ref{53}) we see that $f^\uparrow(x)=\left<\eta^x,f\right>=\left<\xi_x,S^{-1}f\right>=S^{-1}f(x)$, while $f^\downarrow(x)=\left<\eta_x,f\right>=\left<\xi_x,Sf\right>=Sf(x)$. Therefore, for all $f\in\Hil$, $f^\downarrow(x)=S^2 f^\uparrow(x)$. It is further possible to check that the operators $S_\eta$ and $S^{\eta}$ are both bounded, so that their domains coincide with the whole $\Hil$. In fact,
$$
S_\eta f=\int_{\Bbb R}\left<\eta_x,f\right>\eta_x=\int_{\Bbb R}\left<S\xi_x,f\right>S\xi_x=S\int_{\Bbb R}\left<\xi_x,Sf\right>\xi_x=S^2f.
$$
Then $S_\eta=S^2$, and in a similar way we find that $S^\eta=S^{-2}$. Hence, as stated before, these operators are the inverse one of the other, and they are both everywhere defined. Of course, Proposition \ref{prop2} implies that $f^\uparrow(x),f^\downarrow(x)\in\Lc^2(\Bbb R)$.

Formulas (\ref{57}) and (\ref{58}) can also be explicitly checked. Finally, $S$ relates  the generalized vectors $\rho_{kq}$, $\Psi_{kq}$ and $\Phi_{kq}$ introduced before and in the Appendix 1, as follows:
\be
\rho_{kq}=S\Psi_{kq}=S^{-1}\Phi_{kq}.
\label{65}\en
In particular, this means that $\Phi_{kq}=S^2\Psi_{kq}=S_\eta\Psi_{kq}$, which is in agreement with formula (\ref{521}).
We conclude that the case of regular BCS fits perfectly in our framework, and all the quantities introduced in Section III can be easily identified.

\section{Conclusions}\label{sectconcl}

We have considered the possibility of extending the $kq$-representation to operators $\hat x$ and $\hat p$ arising from suitable combinations of pseudo-bosonic operators $a$ and $b$. We have shown that, under natural conditions, such an extension is possible and shares similar properties with its standard  version. The price we have to pay is to deal with biorthogonal (in $\Mc$) sets of vectors. This is not a big surprise, since biorthogonality is somehow intrinsic in any physical model described by non self-adjoint Hamiltonians, \cite{baginbagbook,ben,mosta}.

 We have also used this extended $kq$-representation to prove that the overcomplete sets of BCS constructed out of $a$ and $b$ contain two discrete subsets  which are complete in $\Hil$. The particular case of RBCS has been discussed in details.

Next steps in our analysis will include the detailed analysis of explicit examples of the $kq$-representation discussed here, and its possible applications to physical problems, especially in the realm of pseudo-hermitian quantum mechanics. Another interesting aspect, to be considered in more details, has to do with the position and momentum representations of quantum mechanics when the (generalized) position and/or  momentum operator are not self-adjoint.

\section*{Acknowledgements}

The author acknowledges partial support from Palermo University and from G.N.F.M. of the INdAM.

 \appendix

\renewcommand{\theequation}{\Alph{section}.\arabic{equation}}

 \section{\hspace{-.7cm}ppendix 1:  The kq-representation for $\hat x_0$ and $\hat p_0$}

The relevance of the $kq$-representation in many-body physics has
been established since its first appearances, \cite{zak}, becoming later also an interesting mathematical subject of interest, \cite{jan,dau2}. To keep the paper self-contained, we list here  few definitions and results on the $kq$-representation, referring to the cited papers, and to references therein, for more details.

Let $\hat x_0$ and $\hat p_0$ be the self-adjoint position and momentum operators, satisfying $[\hat x_0,\hat p_0]=i\1$. We define the unitary operators $\tau_1=e^{i\alpha\hat x_0}$ and $\tau_2=e^{-i\alpha\hat p_0}$, where $\alpha^2=2\pi L$, for some $L=1,2,3,\ldots$. Then $[\tau_1,\tau_2]=0$. Also, if $f(x)\in\Lc^2({\Bbb R})$, then
$$
\tau_1f(x)=e^{i\alpha x}f(x),\qquad \tau_2f(x)=f(x-\alpha).
$$
The $kq$-representation makes use of the fact that, since $\tau_1$ and $\tau_2$ commute, they can be diagonalized simultaneously. However, the common eigenstates,

\be \rho_{kq}(x)=\frac{1}{\sqrt{\alpha}}
\sum_{n\in\Bbb{Z}}e^{ikn\alpha}\delta(x-q-n\alpha), \quad\quad k,q\in[0,\alpha[,
 \label{a1} \en
do not belong to $\Lc^2({\Bbb R})$, as it is clear, but can still be thought as elements of the space of tempered distributions $\Sc'(\Bbb R)$. For this reason, and since
\be
\tau_1\rho_{kq}(x)=e^{i\alpha q}\rho_{kq}(x),\qquad \tau_2\rho_{kq}(x)=e^{-i\alpha k}\rho_{kq}(x),
\label{a2}\en
they are called {\em generalized eigenstates} of $\tau_1$ and $\tau_2$. If
$\xi_x$ is the generalized eigenvector of the position operator
$\hat x_0$, $\hat x_0\,\xi_x=x\,\xi_x$, we can write $\rho_{kq}(x)$
as $\rho_{kq}(x)=<\xi_x,\rho_{kq}>$. If we now put $$\Box=\{(k,q)\in{\Bbb R}^2: k,q\in[0,\alpha[\},$$
it is possible to check that
\be
\int\int_{\Box} \overline{\rho_{kq}(x)}\rho_{kq}(x')dk\,dq=\delta(x-x'),
\label{a3}\en
and that
\be
\int_{\Bbb R} \overline{\rho_{kq}(x)}\rho_{k'q'}(x)dx=\delta(k-k')\delta(q-q').
\label{a4}\en
In terms of the $\rho_{kq}$ the previous formulas can be written as:
\be
\left<\rho_{kq},\rho_{k'q'}\right>=\delta(k-k')\delta(q-q'), \quad \int\int_{\Box} |\rho_{kq}\left>\right<\rho_{kq}|dk\,dq=\1,
\label{a5}\en
and
\be
\tau_1\rho_{kq}=e^{i\alpha q}\rho_{kq},\qquad \tau_2\rho_{kq}=e^{-i\alpha k}\rho_{kq}.
\label{a6}\en
As already said, the states $\rho_{kq}(x)$ can also be used  to define a new
representation of the wave functions by means of the integral
transform $Z: \Lc^2(\Bbb{R})\rightarrow \Lc^2(\Box)$, defined as follows: \be
h(k,q):=(ZH)(k,q):=\int_{\Bbb{R}}dx
\overline{\rho_{kq}(x)}H(x), \label{a7} \en for all
functions $H(x)\in \Lc^2(\Bbb{R})$. The result is a function
$h(k,q)\in \Lc^2(\Box)$.

Of course, to be more rigorous, $Z$ should be defined first on the functions of
${\cal C}_o^\infty(\Bbb{R})$ and then extended to $\Lc^2(\Bbb{R})$
using its continuity. This is the approach considered in \cite{jan}, to whom we refer for more mathematical details.

 \appendix

 \section{\hspace{-.7cm}ppendix 2:  formula (\ref{58})}

To check (\ref{58}) it is convenient to introduce, other than $\F_\eta$ and $\F^\eta$, the sets of generalized eigenstates  of $\hat p$ and $\hat p^\dagger$. We first assume that a tempered distribution $\theta_p$ exists such that
\be
\hat p\,\theta_p=p\,\theta_p,
\label{ap1}
\en
for all $p\in\Bbb R$, and we call $\F_\theta$ the set of all these distributions. Now we assume that $\F_\theta$ is well-behaved. Hence a second set $\F^\theta$ of tempered distributions $\theta^p$ also exists, $p\in \Bbb R$, such that  $\left<\theta_p,\theta^q\right>=\delta(p-q)$ and $\int_{\Bbb R}dp|\theta_p\left>\right<\theta^p|=\int_{\Bbb R}dp|\theta^p\left>\right<\theta_p|=\1$. Here these resolutions hold in $\Mc$. We can check that $\theta^p$ is a generalized eigenstate of $\hat p^\dagger$: $\hat p^\dagger\,\theta^p=p\,\theta^p$. Moreover, as in (\ref{53}), we put
\be
\hat f^\uparrow(p):=\left<\theta^p,f\right>,\qquad \hat f^\downarrow(p):=\left<\theta_p,f\right>,
\label{ap2}\en
and we deduce that
\be
\hat f^\uparrow(p)=\int_{\Bbb R}\left<\theta^p,\eta_x\right>\hat f^\uparrow(x),\qquad \hat f^\downarrow(p)=\int_{\Bbb R}\left<\theta_p,\eta^x\right>\hat f^\downarrow(x).
\label{ap3}\en
It is natural to interpret $\hat f^{\uparrow,\downarrow}(p)$ as the Fourier transform of $f^{\uparrow,\downarrow}(x)$. For that, we put
\be
\left<\theta^p,\eta_x\right>=\left<\theta_p,\eta^x\right>=\frac{1}{\sqrt{2\pi}}\,e^{-ipx},
\label{ap4}\en
which is in agreement with the fact that $\left<\theta^p,\theta_q\right>=\delta(p-q)$:
$$
\left<\theta^p,\theta_q\right>=\int_{\Bbb R}\left<\theta^p,\eta_x\right>\left<\eta^x,\theta_q\right>\,dx=\frac{1}{2\pi}\int_{\Bbb R}e^{i(q-p)x}\,dx=\delta(p-q).
$$
Now, since $T_2\theta_p=e^{-i\alpha p}\theta_p$ and $\int_{\Bbb R}dp|\theta_p\left>\right<\theta^p|=\1$, using (\ref{ap4}) we have
$$
T_2\eta_x=T_2\left(\int_{\Bbb R}\left<\theta^p,\eta_x\right>\theta_p\,dp\right)=\int_{\Bbb R}\left<\theta^p,\eta_x\right>(T_2\theta_p)\,dp=\frac{1}{\sqrt{2\pi}}\int_{\Bbb R}e^{-i(x+\alpha) p}\theta_p\,dp=
$$
$$
=\int_{\Bbb R}\left<\theta^p,\eta_{x+\alpha}\right>\theta_p\,dp=\eta_{x+\alpha},
$$
which is true at least when $T_2$ can be moved inside the integral. This is the case, for instance,  in Section \ref{sectRBCS}. The other formula in (\ref{58}) can be proved in a similar way.

\vspace{2mm}

{\bf Remark:--} As in Section \ref{sectRBCS}, it is possible to prove the existence of $\theta_p$ and $\theta^p$ as in the case of regular $\D$-PBs. When $\D$-PBs are not regular, as in the Example discussed in Section \ref{sectkqrepr}, $\theta_p$ and $\theta^p$ should be explicitly deduced and the conditions required to these vectors should be checked. For the  non self-adjoint shifted harmonic oscillator, \cite{baginbagbook,bagharmosc}, this could be explicitly done easily. In other words, for this model both the coordinate and the momentum representations associated to the non self-adjoint operators $\hat x$ and $\hat p$ can be introduced.

\end{document}